\documentclass[conference]{IEEEtran}

\long\def\symbolfootnote[#1]#2{\begingroup%
\def\thefootnote{\monster \fnsymbol{footnote}}\footnote[#1]{#2}\endgroup} 

\usepackage{cite}
\usepackage{graphicx}
\usepackage{url}
\usepackage{amsmath}
\usepackage{latexsym,epsfig,amssymb}

\newcommand{\BEQA}{\begin{eqnarray}}
\newcommand{\EEQA}{\end{eqnarray}}

\newcommand{\figSize}{3in}
\newcommand{\figSpace}{\vspace{-0.2in}}
\newcommand{\figCaptionSpace}{\vspace{-0.1in}}

\newcommand{\sysState}{{\bf X}}
\newcommand{\D}[2]{\frac{\partial #1}{\partial #2}}

\newtheorem{theorem}{Theorem}
\newtheorem{lemma}[theorem]{Lemma}

\newenvironment{definition}[1][Definition]{\begin{trivlist}
\item[\hskip \labelsep {\bfseries #1}]}{\end{trivlist}}

\begin{document}

\title{Designing ISP-friendly Peer-to-Peer Networks\\ Using Game-based Control}


\author{\IEEEauthorblockN{Vinith Reddy\IEEEauthorrefmark{1},
Younghoon Kim\IEEEauthorrefmark{2},
Srinivas Shakkottai\IEEEauthorrefmark{1} and
A.L.Narasimha Reddy\IEEEauthorrefmark{1}}
\IEEEauthorblockA{\IEEEauthorrefmark{1}Dept. of ECE, Texas A\&M University\\
Email: \{vinith\_reddy, sshakkot, reddy\}@tamu.edu}
\IEEEauthorblockA{\IEEEauthorrefmark{2}Dept. of CS, Korea Advanced Institute of Science and Technology\\ 
Email: kyhoon@gmail.com}}


\maketitle

\begin{abstract} The rapid growth of peer-to-peer (P2P) networks in the past few years has brought with it increases in transit cost to Internet Service Providers (ISPs), as peers exchange large amounts of traffic across ISP boundaries.  This ISP oblivious behavior has resulted in misalignment of incentives between P2P networks---that seek to maximize user quality---and ISPs---that would seek to minimize costs.  Can we design a P2P overlay that accounts for both ISP costs as well as quality of service, and attains a desired tradeoff between the two?  
We design a system, which we call MultiTrack, that consists of an overlay of multiple \emph{mTrackers} whose purpose is to align these goals. mTrackers split demand from users among different ISP domains while trying to minimize their individual costs (delay plus transit cost) in their ISP domain. We design the signals in this overlay of mTrackers in such a way that potentially competitive individual optimization goals are aligned across the mTrackers.  The mTrackers are also capable of doing admission control in order to ensure that users who are from different ISP domains have a fair chance of being admitted into the system, while keeping costs in check. We prove analytically that our system is stable and achieves maximum utility with minimum cost. Our design decisions and control algorithms are validated by Matlab and ns-2 simulations.
\end{abstract}



\section{Introduction}
The past few years have seen the rapid growth of content distribution over the Internet, particularly using  peer-to-peer (P2P) networks.  Recent studies estimate that 35-90\% of bandwidth is consumed by P2P file-sharing
applications, both at the edges and even within the core~\cite{fraleigh:traffic,gummadi:traffic}.   The use of P2P networks for media delivery is expected to grow still further, with the proliferation of legal applications (\emph{e.g.} Pando Networks \cite{pando}) that use P2P as a core technology.

\begin{figure}[htbp]
\vspace{-0.1in}
\begin{center}
\includegraphics[width=2.5in]{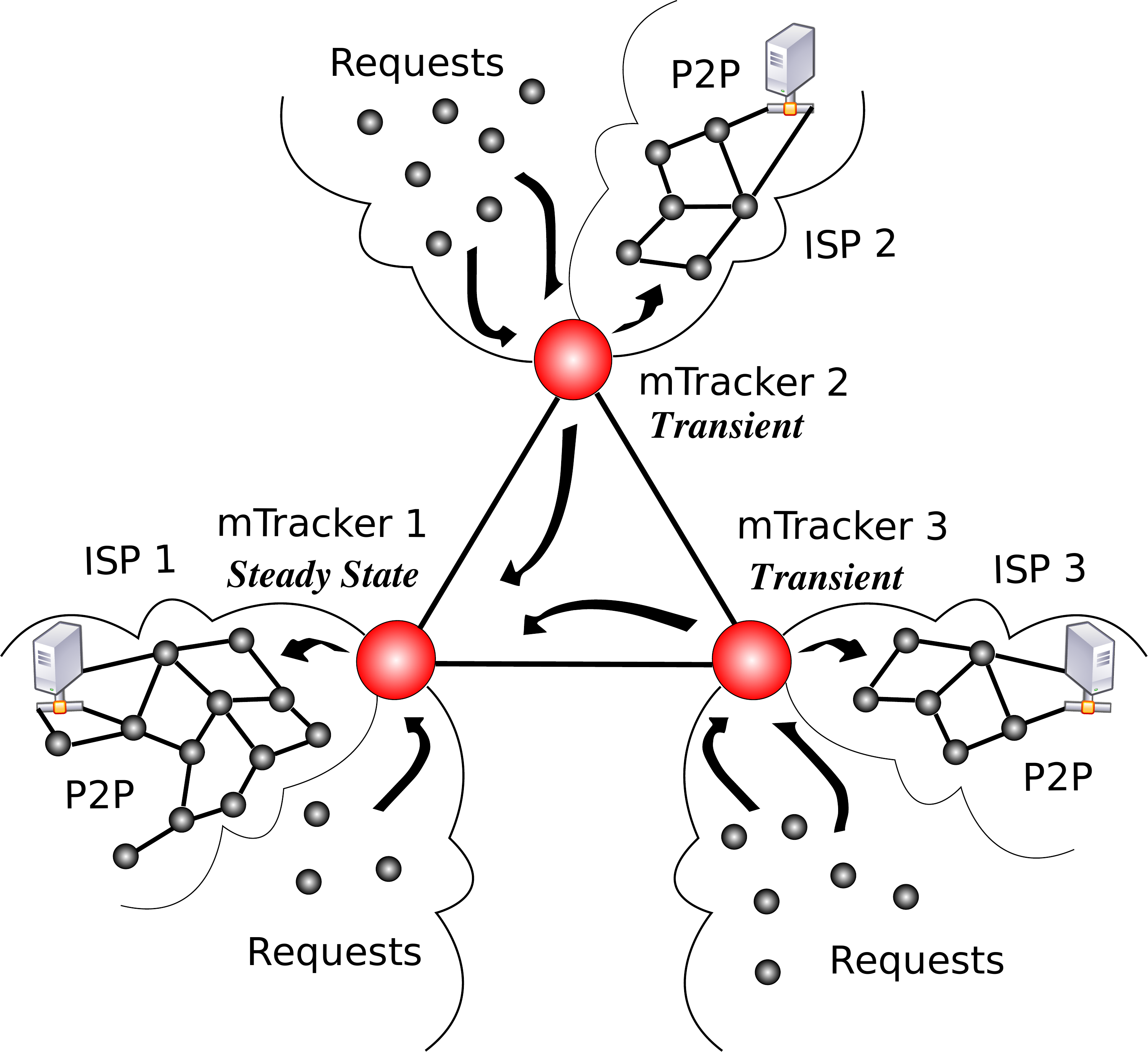}
\figCaptionSpace
 \caption{The MultiTrack architecture.  Multiple trackers, each following individual optimizations, achieve an optimal delay-cost tradeoff.}
\label{fig:multitrack}
\end{center}
\figSpace
\end{figure}

While most P2P systems today possess some form of network resource-awareness, and attempt to optimally utilize the system resources, they are largely agnostic to Internet Service Providers' (ISP) concerns such as traffic management and costs.  This ISP-oblivious nature of P2P networks has hampered the ability of system participants to correctly align incentives.  Indeed, the recent conflicts between ISPs and content providers, as well as efforts by some ISPs such as Comcast to limit P2P traffic on their networks~\cite{comcast:blocking}, speak in part to an inability to align interests correctly.  Such conflicts are particularly critical as P2P becomes an increasingly prevalent form of content distribution~\cite{stats}.

A traditional BitTorrent system \cite{coh03} has elements called \emph{Trackers} whose main purpose is to enable peers to find each other.  The BitTorrent Tracker randomly assigns a new (entering) user a set of peers that are already in the system to communicate with.  This system has the disadvantage that if peers who are assigned to help each other are in the domains of different ISPs, they would cause significant transit costs to the ISPs due to the inter-ISP traffic that they generate.  However, if costs are reduced by forcing traffic to be local, then the delay performance of the system could suffer.  Recent work such as \cite{anja,p4p:sigcomm,chobus08} has focused on cost in terms of load balancing and localizing traffic, and developed heuristics to attain a certain quality of service (QoS).  For example, P4P \cite{p4p:sigcomm} develops a framework to achieve minimum cost (optimal load balancing) among ISP links, but its BitTorrent implementation utilizes the heuristic that $30\%$ of peers declared to each requesting user should be drawn from ``far away ISPs'' in order to attain a good QoS. 

This leads us to the fundamental question that we attempt to answer in this paper: \emph{Can we develop a distributed delay and cost optimal P2P architecture?}  We focus on developing a provably optimal price-assisted architecture called MultiTrack, that would be aware of the interaction between delay and cost.  The idea is to understand that while the resources available with peers in different ISP domains should certainly be used, such usage comes at a price.  The system must be able to determine the marginal gain in performance for a marginal increase in cost.  It would then be able to locate the optimal point at which to operate. 

The conceptual system\footnote{We presented some basic ideas on the system as a poster \cite{RedKim09}.} is illustrated in Figure \ref{fig:multitrack}.  The system is managed by a set of \emph{mTrackers}.  Each mTracker is associated with a particular ISP domain.  The mTrackers are similar to the Trackers in BitTorrent \cite{coh03}, in that their main purpose is to enable peers to find each other.  However, unlike BitTorrent, the mTrackers in MultiTrack form an overlay network among themselves.  The purpose of the overlay network is to provide  multi-dimensional actions to the mTrackers.  In Figure  \ref{fig:multitrack}, mTracker 1 is in steady state (wherein the demand on the mTracker is less than the available capacity \cite{p2p:gdv}), which implies that it has spare capacity to serve requests from other mTrackers.  Consider mTracker 2 which is in transient state (wherein the demand on the mTracker is more than the available capacity \cite{p2p:gdv}).  When a request arrives, it can either assign the requester to its own domain at essentially zero cost, or can forward the user to mTracker 1 and incur a cost for doing so.   However, the delay incurred by forwarded users would be less as mTracker 1 has higher capacity.  Thus, mTracker 2 can trade-off cost versus delay performance by forwarding some part of its demand.

Each mTracker uses \emph{price assisted decision making} by utilizing dynamics that consider the marginal payoff of forwarding traffic to that of retaining traffic in the same domain as the mTracker.  Several such rational dynamics have been developed in the field of \emph{game theory} that studies the behavior of selfish users.  
We present our system model with its simplifying assumptions in Section~\ref{sec:system}.  We then design a system in which the actions of these mTrackers, each seeking to maximize their own payoffs, actually result in ensuring lowest cost of the system as a whole.  The scheme, presented in Section~\ref{sec:mTracker}, involves \emph{implicit learning} of capacities through probing and backoff through a rational control scheme known as replicator dynamics \cite{hofsig98,broneu50}. We present a game theoretic framework for our system in Section~\ref{sec:background} and show using Lyapunov techniques that the vector of split probabilities converges to a provably optimal state wherein the total cost in terms of delay and traffic-exchange is minimized.  Further, this state is a Wardrop equilibrium \cite{Wardrop52}.

We then consider a subsidiary problem of achieving fair division of resources among different mTrackers through admission control in Section~\ref{sec:mTrAdmission}.  The objective here is to ensure that some level of fairness is maintained among the users in different mTracker domains, while at the same time ensuring that the costs in the system are not too high.  Admission control implies that not all users in all domains would be allowed to enter the system, but it should be implemented in a manner that is fair to users in different mTracker domains.  The mTracker takes admission control decisions based on the marginal disutility caused by users to the system.  Users interested in the file would approach the mTracker that would decide whether or not to admit the user into the system.  We show that our mTracker admission control optimally achieves fairness amongst users, while maintaining low system cost.  Note that switching off admission control would still imply that the total system cost is minimized by mTrackers, but this could be high if the offered load were high. 

We simulate our system both using Matlab simulations in Section~\ref{sec:matlab_sims} to validate our analysis, as well as ns-2 simulations in Section \ref{sec:ns2_sims} to show a plausible implementation of the system as a whole.  The simulations strongly support our architectural decisions.  We conclude with ideas on the future in Section \ref{sec:conclusion}.  

\section{Related Work}\label{sec:related}

There has been much recent work on P2P systems and traffic management, and we provide a discussion of work that is closely related to our problem.  Fluid models of P2P systems, and the multi-phase (transient/steady state) behavior has been developed in \cite{p2p:gdv,qiusri04}.  The results show how supply of a file correlates with its demand, and it is essentially transient delays that dominate.     

Traffic management and load balancing have become important as P2P networks grow in size.  There has been work on traffic management for streaming traffic \cite{setapo07,sigm:chiang,sigm:chen}.  In particular, \cite{setapo07} focuses on server-assisted streaming, while \cite{sigm:chiang,sigm:chen} aim at fair resource allocation to peers using optimization-decomposition.  

Closest to our setting is work such as \cite{anja,p4p:sigcomm,chobus08}, that study the need to localize traffic within ISP domains.  In \cite{anja}, the focus is on allowing only local communications and optimizing the performance by careful peer selection, while \cite{p4p:sigcomm} develops an optimization framework to balance load across ISPs using cost information.  A different approach is taken in \cite{chobus08}, wherein peers are selected based on inputs on nearness provided by CDNs (if a CDN directs two peers to the same cache, they are probably near by).  

Pricing and market mechanisms for P2P systems are of significant interest, and work such as \cite{market:conext08} use ideas of currency exchange between peers that can be used to facilitate file transfers.  The system we design uses prices between mTrackers that map to real-world costs of traffic exchange, but do not have currency exchanges between peers which still use BitTorrent style bilateral barter.

\section{The MultiTrack System}\label{sec:system}

MultiTrack is a hybrid P2P network architecture similar to \emph{BitTorrent} \cite{coh03, bittorrent} in many ways, and we first review some control elements of BitTorrent.  In the BitTorrent architecture a file is divided into multiple chunks, and there exists at least one  \emph{Tracker} for each file that keeps track of peers that contain the file in its entirety (such peers are called \textit{seeds}) or some chunks of it (such peers are called \textit{downloaders}).  A new peer that wants to download a file needs to first locate a Tracker corresponding to the file.  Information about Trackers for a file (among other information) is contained in .torrent files, which are hosted at free servers.  Thus, the peer downloads the .torrent file, and locates a Tracker using this file.  

When a peer sends a request to a Tracker corresponding to the file it wants, the Tracker returns the addresses of a set of peers (\textit{seeds and downloaders}) that the new peer should contact in order to download the file. The peer then connects to a subset of the given peers and downloads chunks of the file from them. While downloading the file, a peer sends updates to the Tracker about its download status (number of chunks uploaded and downloaded).  Hence, a tracker knows about the state of each peer that is present in its peer cloud (or \textit{swarm}). 

The MultiTrack architecture consists of BitTorrent-like trackers, which we call \textbf{mTrackers}.  We associate one or more mTrackers to each ISP, with each mTracker controlling access to its own peer cloud.  Note that all these mTrackers are identified with the \emph{same} file.  Unlike BitTorrent Trackers, mTrackers are aware of each other and form an overlay network among themselves. Each mTracker consists of two different modules:
\begin{enumerate}  
\item \textbf{Admission control}: Unlike the BitTorrent tracker which has no control over admission decisions of peers, the mTracker can decide whether or not to admit a particular peer into the system. Once admitted, the peer is either served locally or is forwarded to a different mTracker based on the decision taken by the mTracker.
\item \textbf{Traffic management}: This module of the mTracker, takes a decision on whether to forward a new peer into its own peer cloud (at relatively low cost, but possibly poor delay performance) or to another mTracker (at higher cost, but potentially higher performance).
\end{enumerate} 

The rationale behind this architecture is as follows.  At any time, a peer cloud has a capacity associated with it, based on the maximum upload bandwidth of a peer in the cloud and the total number of chunks present at all the peers in the cloud (seeds and downloaders).  In general, a peer-cloud has two phases of operation \cite{p2p:gdv}---a \emph{transient} phase where the available capacity is less than the demand (in other words, not enough peers with a copy of the file), and a \emph{steady state} phase, where the available capacity is greater than the capacity required to satisfy demand.  Thus, a peer cloud can be thought of as a server with changing service capacity.   We balance load among different peer clouds located in different ISPs, taking into account the transit cost associated with traffic exchange. 

We assume \emph{time scale separation} between the two modes--- traffic management and admission control, of the mTracker.  Our assumption is that the capacity of a P2P system remains roughly constant over intervals of time, with capacity changes seen at the end of these time periods.  We divide system dynamics into three time scales:
\begin{enumerate}
\item \textbf{Large:} The capacity of the peer cloud associated with each  mTracker changes at this time scale.
\item \textbf{Medium:} mTrackers take admission control decisions at this time scale.  They could increase or decrease the number of admitted peers based on feedback from the system.  We will study dynamics at this time scale in Section \ref{sec:mTrAdmission}.
\item \textbf{Small:} mTrackers split the demand that they see among the different options (mTrackers visible to them) at this time scale.  Thus, they change the probability of sending peers to their own peer-cloud or to other mTrackers at this time scale.  We study these dynamics in Section \ref{sec:mTracker}.
\end{enumerate}
Note that a medium time unit comprises of many small time units and a large time unit comprises of many medium time units.  The artifice of splitting dynamics into these time scales allows us to design each control loop while assuming that certain system parameters remain constant during the interval.  In the following sections, we present the design and analysis of our different system components.

\section{mTracker: Traffic Management}\label{sec:mTracker}

The objective of the mTracker's \emph{Traffic Management} module is to split the demand that it sees among the different options (other mTrackers, and its own peer cloud) that it has. Since each mTracker is associated with a different ISP domain, it would  like to minimize the cost seen by that ISP, and yet maintain a good delay performance for its users. 

As mentioned in the last section, peer-clouds can be in either \emph{transient} or \emph{steady-state} based on whether the demand seen is greater than or less than the available capacity.  An mTracker in the transient mode would like to offload some of its demand, while mTrackers in the steady-state mode can accept load.  Thus, each mTracker $j$ in the transient mode maintains a split probability vector $\hat{\vec{y}}_j = [\hat{y}^1_j \dots \hat{y}^Q_j]$, where $Q$ is the total number of mTrackers, and some of the $y^i_j$ could be zero.  We assume that the demand seen by mTracker $j$ is a Poisson process of rate $x_j$.  Thus, splitting traffic according to $\hat{\vec{y}}_j$ would produce $Q$ Poisson processes, each with rate $x^i_j\triangleq y^i_j x_j$ ($i=1,...Q$).
  
Now, each mTracker in the steady-state mode can accept traffic from mTrackers that are transient.  It could, of course, prioritize or reserve capacity for its own traffic; we assume here that it does so, and the balance capacity available (in users served per unit time) of this steady state mTracker is $C^i$.  Then the demand seen at each such mTracker $i$ is the sum of Poisson processes that arrive at it, whose rate is simply $\sum_{l=1}^Q x^i_l.$  We assume that delay seen by each peer sent to mTracker $i$ is convex increasing in load, and for illustration use the M/M/1 delay function
\BEQA\label{eqn_mm1_delay}
\frac{1}{C^i-\sum_{l=1}^Q x^i_l}.
\EEQA 
Note that, peers from different transient mTrackers are not allowed to communicate with each other at the steady state mTracker to which they are forwarded. Thus, a peer that is forwarded from one ISP domain to another is only allowed to communicate with peers located in that ISP domain.
 
Now, the steady state mTrackers are disinterested players in the system, and would like to minimize the total delay of the system.   They could charge an additional price that would act as a congestion signal to mTrackers that forward traffic to them.  Such a congestion price should reflect the ill-effect that increasing the load by one mTracker has on the others.  What should such a price look like?  Now, consider the expression
\BEQA
D(z)=\frac{1}{ C^i - z^i},
\EEQA
which is the general form of the delay seen by each user at mTracker $i$.  The elasticity of delay with arrival rate $z_i$
\BEQA\label{eqn:elasticity}
\frac{\partial D(z^i)}{\partial z^i}\frac{z^i}{D(z^i)} = \frac{z^i}{C^i-z^i}.
\EEQA
The elasticity gives the fractional change in delay for a fractional change in load, and can be thought of as the cost of increasing load on the users.  In other words, if the load is increased by any one mTracker, \emph{all} the others would also be hurt by this quantity.  Expressing the above in terms of delay (multiplying by total delay) to ensure that all units are in delay, the elasticity per unit rate per unit time at mTracker $i$ is just
\BEQA\label{eqn:elasticity_cost}
\frac{\sum_{l=1}^Qx^i_l}{( C^i - \sum_{l=1}^Qx^i_l)^2}.
\EEQA
The above quantity represents the ill effect that increasing the load per unit time has on the delay experienced on all users at mTracker $i$.  The delay cost (\ref{eqn_mm1_delay}) is the disutility for using the mTracker, while the congestion cost (\ref{eqn:elasticity_cost}) is the disutility caused to others using the mTracker.  The mTracker can charge this price to each mTracker that forward peers to it.  

Since mTrackers belong to different ISP domains, forwarding demand from one mTracker to the other is not free.  We assume that the transit cost per unit rate of forwarding demand from mTracker $j$ to mTracker $i$ is $p^i_j$.  Thus, the payoff of mTracker $j$ due to forwarding traffic to mTracker $i$ per unit rate per unit time is given by the sum of transit cost $p^i_j$ with the delay cost (\ref{eqn_mm1_delay}) and congestion price (\ref{eqn:elasticity_cost}), which yields a total payoff per unit rate per unit time of
\BEQA
\frac{1}{ C^i - \sum_{l=1}^Qx^i_l} + p^i_j + \frac{\sum_{l=1}^Qx^i_l}{( C^i - \sum_{l=1}^Qx^i_l)^2}.
\EEQA
The mTracker would like as \emph{small} a payoff as possible.

In the next subsections we will develop a population game model for our system, and show how rational dynamics when coupled with the payoff function given above naturally results in minimizing the total system cost (delay cost plus transit cost).  A good reference on population games is \cite{Sandholm01}.

\subsection{MultiTrack Population Game}\label{sec:background}
A \emph{population game} $\mathcal{G}$, with a set $\mathcal{Q}=\{1, ... ,Q\}$ of non-atomic populations of players is defined by the following entities:
\begin{enumerate}
\item a mass, $x_j \quad \forall j \in \mathcal{Q}$,
\item a strategy or action set, $\mathcal{S}_j = \{1, ... , S_j \} \quad \forall j \in \mathcal{Q}$ and
\item a payoff, $F^i_j \quad \forall j \in \mathcal{Q}$ and $\forall i \in \mathcal{S}_j$. 
\end{enumerate}
By a non-atomic population, we mean that the contribution of each member of the population is infinitesimal. 

In the \emph{MultiTrack Game} each mTracker is a player and the options available to each mTracker are other mTrackers' peer cloud or its own peer cloud. Let $\vec{x}= [x_1,\dots x_Q]$ be the total load vector of the system at the small time scale, where $x_j \quad \forall\ j \in \mathcal{Q}$  is the total arrival rate of new peer requests (or mass) at mTracker $j$.  A strategy distribution of an mTracker $j \in \mathcal{Q}$ is a split of its load $x_j$ among different mTrackers including itself, represented as $\vec{x}_j = [x^1_j \dots x^Q_j]$, where $\sum^Q_{i=1}x^i_j = x_j$. 
If a mTracker $j$ is not connected to mTracker $i$ (or if it does not want to use mTracker $i$), then the rate $x^i_j = 0$.
We denote the vector of strategies being used by all the mTrackers as $\sysState = [\vec{x}_1 \dots \vec{x}_Q]$. The vector $\sysState$ represents the state of the system and it changes continuously with time.  Let the space of all possible states of a system for a given load vector be denoted as $ \mathbb X$, i.e $\sysState \in \mathbb X$. 

The payoff (per unit rate per unit time) of forwarding requests from mTracker $j$ to $i$, when the state of the system is $\sysState$ is denoted by $F_j^i(\sysState) \in \mathbb R$ and is assumed to be continuous and differentiable.  As developed above this payoff is 
\BEQA \label{eqn_F_i_j}
F_j^i(\sysState) = \frac{1}{ C^i - \sum_{l=1}^Qx^i_l} + p^i_j + \frac{\sum_{l=1}^Qx^i_l}{( C^i - \sum_{l=1}^Qx^i_l)^2}.
\EEQA
Recall that mTrackers want to keep payoff \emph{as small as possible}.

A commonly used concept in non-cooperative games in the context of infinitesimal players, is the Wardrop equilibrium \cite{Wardrop52}. Consider any strategy distribution $\vec{x}_j=[x^1_j,...,x^{S_j}_j]$.  There would be some elements which are non-zero and others which are zero.  We call the strategies corresponding to the non-zero elements as the \emph{strategies used by population $j$}.
\begin{definition} {\bf 1}
A state ${\bf \hat{X}}$ is a
Wardrop equilibrium if for any population $j \in \mathcal{Q}$,
all strategies being used by the members of $j$ yield the same marginal payoff
to each member of $j$, whereas the marginal payoff that would
be obtained by members of $j$ is higher for all strategies not used by population $j$.
\end{definition}
In the context of our MultiTrack game the above definition of Wardrop equilibrium is characterized by the following relation:
\BEQA
F_j^r( {\bf \hat{X}} ) \leq F_j^i ( {\bf \hat{X}})\ \quad \forall\ r\in \mathcal{\hat{Q}}_j\mbox{ and } i\in \mathcal{Q}
\nonumber
\EEQA
Where $\mathcal{\hat{Q}}_j\subset \mathcal{Q}$ is the set of all mTrackers used by population $j$ in a strategy
distribution ${\bf \hat{\vec{x}}_j}$.

The above concept refers to an \emph{equilibrium condition}; the question arises as to how the system actually arrives at such a state.  One model of population dynamics is \emph{Replicator Dynamics} \cite{hofsig98}.  The rate of increase of ${\dot{x}^i_j}/{x^i_j}$ of the strategy $i$
is a measure of its evolutionary success.  Following ideas of Darwinism,
we may express this success as the difference in fitness $F^i_j({\bf X})$ of the strategy $i$
and the average fitness $\sum_{r=1}^{Q} x^r_j F^r_j({\bf X})/x_j$ of the population $j$.  Then we obtain
\BEQA
\frac{\dot{x}^i_j}{x^i_j}= \mbox{average fitness - fitness of $s$}.\nonumber
\EEQA
Then the dynamics used to describe changes in the mass of population $j$ playing strategy $s$
is given by
\vspace{-0.1in}
\BEQA\label{eqn_RepDyn}
\dot{x}^i_j=x^i_j\Big(\frac{1}{x_j}\sum_{r=1}^{Q} x^r_j F^r_j({\bf X}) - F^i_j({\bf X})\Big).
\EEQA
The above expression implies that a population would increase the mass of a successful strategy and decrease the mass of a
less successful one.  It is called the replicator equation after the tenet ``like begets like''.  Note that the total mass of the population $j$ is constant.  We design our mTracker Traffic Management module around Replicator Dynamics (\ref{eqn_RepDyn}).
  
\subsection{Convergence of mTracker dynamics}

We define the total cost in the system to be the sum of the total system delay plus the total transit cost.  In other words, we have weighted delay costs and transit costs equally when determining their contribution to the system cost.  We could, of course, use any convex combination of the two without any changes to the system design.  Hence using the M/M/1 delay model at each tracker and adding transit costs, the total system cost when the system is in state $\sysState$ is given as:
\BEQA \label{eqn_Lyap}
\mathcal{C}(\sysState) = \sum_{i=1}^Q\left\{ \frac{\sum_{r=1}^Qx^i_r}{ C^i - \sum_{l=1}^Qx^i_l} + \sum_{r=1}^Q p^i_rx^i_r \right\}.
\EEQA
Note that the cost is convex and increasing in the load.  We will show that the above expression acts as a Lyapunov function for the system.
 
\begin{theorem} \label{thm:repStable}
The system of mTrackers that follow \textit{replicator dynamics} with payoffs given by (\ref{eqn_F_i_j}) is globally asymptotically stable.
\end{theorem}
\begin{proof}
We prove the system stability using Lyapunov Theory with $\mathcal{C}(\sysState)$ defined in (\ref{eqn_Lyap}) as the Lyapunov function.  From (\ref{eqn_F_i_j}) and (\ref{eqn_Lyap}), $\D{\mathcal{C}}{x_j^i} = F_j^i(\sysState)$, hence
\BEQA
\dot{\mathcal{C}}(\sysState) = \sum_{i=1}^Q \sum_{j=1}^Q \D{\mathcal{C}}{x_j^i} \dot{x}_j^i = \sum_{i=1}^Q \sum_{j=1}^Q F_j^i(\sysState)\dot{x}_j^i, \label{eqn:C_dot}
\EEQA
Now, let $\tilde{\mathcal{X}}$ be the set of states such that,
$$\dot{\mathcal{C}}(\sysState) = 0, \forall\  \sysState \in \tilde{\mathcal{X}}$$  
From (\ref{eqn:C_dot}) it is evident that $\dot{\mathcal{C}}(\sysState) = 0$,  if:
\BEQA
F_j^i(\sysState) = 0 \mbox{ or} \hspace{1in} \\
\left(\dot{x}_j^i = 0\right) \Rightarrow \frac{1}{x_j} \sum_{r=1}^Qx_j^r F_j^r = F_j^i \quad \forall\ i,j \in Q
\EEQA
Thus, $\tilde{\mathcal{X}}$ is the set of equilibrium states of replicator dynamics.  We will show that $\dot{\mathcal{C}}(\sysState)<0 \ \forall\  X \notin \tilde{\mathcal{X}}.$
 
From (\ref{eqn:C_dot}) and (\ref{eqn_RepDyn}) we have
\BEQA
\dot{\mathcal{C}}(\sysState) = \sum_{i=1}^Q \sum_{j=1}^Q F_j^i x_j^i\left(\frac{1}{x_j} \sum_{r=1}^Qx_j^r F_j^r - F_j^i \right) \\
= \sum_{j=1}^Q x_j \left(\left(\sum_{i=1}^Q\frac{x_j^i}{x_j} F_j^i\right)^2 - \left(\sum_{i=1}^Q\frac{x_j^i}{x_j}(F_j^i)^2\right)\right)
\EEQA
Since function $f(x) = x^2$ is convex and $\sum_{i=1}^Q \frac{x_j^i}{x_j} =1$, from Jensen's inequality we have, 
$\forall\  X \notin \tilde{\mathcal{X}}$:
\BEQA
\left(\left(\sum_{i=1}^Q\frac{x_j^i}{x_j} F_j^i\right)^2 - \left(\sum_{i=1}^Q\frac{x_j^i}{x_j}(F_j^i)^2\right)\right) & < & 0 \quad \forall\  j \in \mathcal{Q} \nonumber
\EEQA
Thus, $\dot{\mathcal{C}}(\sysState) <  0, \quad \forall\  X \notin \tilde{\mathcal{X}}$ and the system is globally asymptotically stable \cite{khalil96}.
\end{proof}
While replicator dynamics is a simple model, it has a drawback: during the different iterations of replicator dynamics, if the value of $x_j^i$, the rate of forwarding requests from mTracker $j$ to mTracker $i$ becomes zero then it remains zero forever. Thus, a strategy could become extinct when replicator dynamics is used and its stationary point might not be a Wardrop equilibrium.  To avoid this problem, we can use another kind of dynamics called Brown-von Neuman-Nash (BNN) Dynamics, which is described as:
\BEQA 
 \dot{x}^i_q&=&x_q\gamma^i_q-x^i_q\sum_{j=1}^{Q}\gamma_q^j \label{eqn_bnn}\\
\mbox{where, }  \gamma_q^i&=&\max\big\{F^i_q(X) -\frac{1}{x_q}\sum_{i=1}^{Q} x^i_q F^i_q(X),0\big\} \label{eqn_excess}
\EEQA
denotes the excess payoff to strategy $i$ relative to the average payoff in population $q$.  We can show that the system of mTrackers that follow BNN dynamics is globally asymptotically stable in a manner similar to the proof of Theorem~\ref{thm:repStable}.

We have just shown that the total system cost acts as a Lyapunov function for the system.  It should not come as a surprise then, that the cost-minimizing state is a Wardrop equilibrium.  We prove this formally in the next subsection.

\subsection{Cost efficiency of mTrackers}\label{subsec:mTrackCostEff}

In previous work on selfish routing  (\emph{e.g.} \cite{roughgarden00}), it was shown that the Wardrop equilibrium does not result in efficient system performance. This inefficiency is referred to as the \emph{price of anarchy}, and it is primarily due to selfish user-strategies.  However, work on population games \cite{Sandholm01} suggests that carefully devised price signals would indeed result in efficient equilibria.  We show now that the Wardrop equilibrium attained by mTrackers is efficient for the system as a whole.

The objective of our system is to minimize the total cost for a given load vector $\vec{x} = [x_1, \dots, x_Q]$. Here the total cost in the system is $\mathcal{C}(\sysState )$ and is defined in (\ref{eqn_Lyap}). This can be represented as the following constrained minimization problem:
\BEQA\label{eqn_sysPrimal}
\min_{\sysState}&\mathcal{C}({\bf X} ) &\\
\mbox{subject to, } & \sum_{i=1}^Qx_j^i = x_j \quad &\forall \ j \in \mathcal{Q} \label{eqn_sysConstraints}\\
& x_j^i \geq  0 & 
\EEQA
The Lagrange dual associated with the above is
\BEQA\label{eqn_sysDual}
\mathcal{L}(\lambda,\sysState) = \max_{\lambda, h}\min_{\sysState}\bigg(\mathcal{C}(\sysState ) &-& \\
\sum_{j=1}^Q \lambda _j\Big(\sum_{i=1}^Qx_j^i -x_j\Big) &-& \sum_{i=1}^Q \sum_{j=1}^Q h_j^i x_j^i \bigg) \nonumber
\EEQA
where $h_j^i \geq 0$ and $\lambda_j$, $\forall \ i,j, \in \mathcal{Q}$ are the dual variables.  Now the above dual problem gives the following Karush-Kuhn-Tucker first order conditions:
\BEQA
\D{\mathcal{L}}{x_j^i}(\lambda,\sysState ^\star) = 0 &  \forall\  i, j \in \mathcal{Q}  \label{eqn_sysKKT} \\
\mbox{and} \nonumber \\ 
h_j^i x_j^{\star i} = 0 & \forall\  i, j \in \mathcal{Q} \label{eqn_KKT2}
\EEQA
where $\sysState ^\star$ is the global minimum for the primal problem (\ref{eqn_sysPrimal}).  Hence, from (\ref{eqn_sysKKT}) we have
\BEQA \nonumber
&\D{\mathcal{C}}{x_j^i}(\sysState ^\star) - \lambda_j \D{(\sum_{i=1}^Q x^{\star_i}_j - x^\star_j)}{x_j^i} + h_j^i\ = \ 0 & \forall\ i,j \in \mathcal{Q} \\
\Rightarrow & \D{\mathcal{C}}{x_j^i}(\sysState ^\star) \ = \ \lambda_j + h_j^i \quad \forall\ i,j \in \mathcal{Q} \label{eqn_KKT}
\EEQA
We know from the definition of payoff (\ref{eqn_F_i_j}) that $ \D{\mathcal{C}}{x_j^i}(\sysState ) = F_j^i(\sysState )$.  Thus from (\ref{eqn_KKT}) we have
\BEQA
F_j^i(\sysState ^\star)  \ = \ \lambda_j + h_j^i & \forall\ i,j \in \mathcal{Q}
\EEQA
From (\ref{eqn_KKT2}), it follows that
\BEQA
F_j^i(\sysState ^\star)  \ = \ \lambda_j & \mbox{ when } x_j^{\star i} > 0 \ \forall\ i,j \in \mathcal{Q} \label{eqn_wdropEq1}\\
\mbox{and} \nonumber \\
F_j^i(\sysState ^\star)  \ = \ \lambda_j + h_j^i & \mbox{ when } x_j^{\star i} = 0 \ \forall\ i,j \in \mathcal{Q} \label{eqn_wdropEq2}
\EEQA
Now, consider the replicator dynamics (\ref{eqn_RepDyn}), at stationary point we have $\dot{x}_j^i = 0$.  Thus,
\BEQA
&\hat{F}_j = \  F_j^i(\hat{X}) \quad \forall\ i,j \in \mathcal{Q} \label{eqn_statPoint} \\
\mbox{or } &\hat{x}_j^i = 0,\nonumber \\
\mbox{where} &\nonumber \\
&\hat{F}_j \triangleq \frac{1}{\hat{x}_j}\sum_{r=1}^Q\hat{x}_j^r F_j^r(\hat{X}) \quad \forall\ j \in \mathcal{Q}, \label{eqn_Favg}
\EEQA
and $\hat{X}$ denotes a stationary point. The above equations imply that for mTracker $j$ the per unit cost of forwarding traffic to all the other mTrackers that it uses is the same. However, for an option $i$ that it does not use, the rate of forwarding $x_j^i$ is $0$ or equivalently, the cost is more than the average payoff.  Finally, we observe that the conditions required for Wardrop equilibrium are identical to the KKT first order conditions (\ref{eqn_wdropEq1})-(\ref{eqn_wdropEq2}) of the minimization problem (\ref{eqn_sysPrimal}) when
 \BEQA
\label{eqn_FeqLambda}\hat{F}_j = \lambda_j \quad \forall\ j \in \mathcal{Q},
\EEQA
which leads to the following theorem.
\begin{theorem}
The solution of the minimization problem in (\ref{eqn_sysPrimal}) is identical to the Wardrop equilibrium of the non-cooperative potential game $\mathcal{G}$. 
\end{theorem}
\begin{proof}
From the above discussion we know that the KKT conditions of (\ref{eqn_sysPrimal}) satisfy the Wardrop equilibrium conditions of the game $\mathcal{G}$. 
%
Thus, to finish this proof all we need to show is that there is no duality gap between the primal (\ref{eqn_sysPrimal}) and the dual (\ref{eqn_sysDual}) problems.   This follows immediately from convexity of the total system cost.
\end{proof}

\section{mTracker: Admission Control}\label{sec:mTrAdmission}

In the previous section we witnessed how each mTracker tries to reduce the cost in its peer cloud by forwarding requests to other mTrackers.  However, minimizing the total delay does not mean that it is bounded.  In order to ensure acceptable delay performance, we provide admission control functionality to each mTracker.  The mTracker takes admission control decisions in the medium time scale; the mTracker's demand splitting is assumed to have converged to yield the lowest cost split at every instant at this time scale. In some ways the admission control mode supplements natural market dynamics---if the delay experienced by requesters were unbearably high, they would simply abort, causing the system to recover.  However, such dynamics might cause large swings in quality over time; the mTracker's admission control precludes the occurrence of such swings.  

We could formulate an admission control problem, enforcing hard constraints on the acceptable system cost, as a convex optimization problem shown below:
\vspace{-0.1in}
\BEQA 
\max_{\vec{x}}&\sum_{j=1}^Q w_j\log{x_j}\label{tTrack_CnvxOpt}\\
\mbox{subject to, } & \mathcal{C}^\star(\vec{x}) \leq \kappa  \\
& x_j \geq 0 \nonumber
\EEQA
where $\vec{x}$ is the load vector and $\mathcal{C}^\star(\vec{x})$ is the minimum value of the optimization problem (\ref{eqn_sysPrimal}) for a given load $\vec{x}$.   We can easily show that the constraint set is (\ref{tTrack_CnvxOpt}) is convex. 
\begin{lemma}
The set of all load vectors $\vec{x}$, satisfying the inequality constraint, $\mathcal{C}^\star(\vec{x}) \leq \kappa$ is convex.
\end{lemma}
\begin{proof}
Let $\vec{x}$ and $\vec{y}$ be two load vectors such that,
\BEQA
\mathcal{C}^\star(\vec{x}) \leq \kappa \mbox{  and  } \mathcal{C}^\star(\vec{y}) \leq \kappa \label{eqn:xystar}
\EEQA
Let $X_{min}$ and $Y_{min}$ be the states, corresponding to load vectors $\vec{x}$ and $\vec{y}$ respectively, which results in minimum cost to the system, i.e., 
\BEQA
\mathcal{C}(X_{min}) = \mathcal{C}^\star(\vec{x}) \mbox{  and  }
\mathcal{C}(Y_{min}) = \mathcal{C}^\star(\vec{y}) \label{eqn:xymin}
\EEQA
Consider,
\BEQA
\mathcal{C}(\alpha X_{min}+(1-\alpha)Y_{min}) \leq \alpha \mathcal{C}(X_{min})+ (1 -\alpha)\mathcal{C}(Y_{min})
\EEQA
the above inequality follows from the convexity of $\mathcal{C}(X)$. \\
Using Eqns(\ref{eqn:xystar}) and (\ref{eqn:xymin}), we get:
\BEQA
\mathcal{C}(\alpha X_{min}+(1-\alpha)Y_{min})\leq & \alpha \mathcal{C}^\star(\vec{x})+(1-\alpha)\mathcal{C}^\star(\vec{y})\\
\leq & \alpha \kappa +(1-\alpha) \kappa \\
\leq & \kappa \label{eqn:C_xmin_ymin}
\EEQA
if $\vec{z} = \alpha \vec{x} + (1 - \alpha )\vec{y}$, then from the definition of $\mathcal{C}^\star$
\BEQA
\mathcal{C}(Z_{min}) &=&\mathcal{C}^\star(\vec{z})
\mathcal{C}^\star(\vec{z}) = \mathcal{C}(Z_{min})
\EEQA
where $Z_{min}$ is the state of the system, corresponding to load $\vec{z}$, when the cost is minimum.

Clearly we can represent any state $Z$, corresponding to the load vector $\vec{z}$, in the form of $\alpha X + (1-\alpha )Y$, and thus it follows from the definition of $\mathcal{C}^\star$ and Eqn(\ref{eqn:C_xmin_ymin}) that:
\BEQA
\mathcal{C}^\star (\alpha \vec{x} + (1- \alpha )\vec{y}) \leq \mathcal{C}(\alpha X_{min} + (1- \alpha )Y_{min}) \leq \kappa
\EEQA
Thus the set is convex.
\end{proof}

If we think of $\sum_jw_j\log{x_j}$ as the total system utility, then $\mathcal{C}^\star(\vec{x})$ is the total system disutility.  Instead of hard constraints on the cost, we relax the problem after the manner of \cite{kel00,shaksri08} to simply ensure that the difference of utility and disutility (the \emph{net utility}) is as large as possible.  
\BEQA\label{eqn_relaxed}
\max_{\vec{x}}&\sum_{j=1}^Q w_j\log{x_j} - \mathcal{C}^\star(\vec{x}) \\
\mbox{subject to, } & x_j \geq 0 \nonumber
\EEQA
A gradient ascent type controller that could be used to solve the above problem is
\BEQA
\dot{x}_j = w_j - x_j \D{\mathcal{C}^\star}{x_j}  \forall\ j \in \mathcal{Q}
\EEQA
To determine the second term above, we use
\BEQA 
\D{\mathcal{C}^\star}{x_j} = \hspace{2.7in}\nonumber \\
\sum_{i \in \mathcal{Q}} \D{\mathcal{C(\sysState ^\star )}}{x^i_j} \D{x^i_j}{x_j} + \sum_{k \in \mathcal{Q}, k \neq j} \sum_{i \in \mathcal{Q}} \D{\mathcal{C(\sysState ^\star )}}{x^i_k} \D{x^i_k}{x_j}  \forall\ j \in \mathcal{Q}\label{eqn:gradC_star}\\
= \sum_{i \in \mathcal{Q}}F^i_j(\sysState ^\star)\D{x^i_j}{x_j} + \sum_{k \in \mathcal{Q}, k \neq j} \sum_{i \in \mathcal{Q}} F^i_k (\sysState ^\star)\D{x^i_k}{x_j} \forall\ j \in \mathcal{Q}\label{eqn:f_star}
\EEQA
When the system is in Wardrop equilibrium ($\sysState ^\star$) all the options $i$ used by $j$ yield the same payoff hence, $F^i_j(\sysState ^\star) = F_j^\star$. For the options $r$ that were not used $\D{x^r_j}{x_j} =0$ at state $\sysState ^\star.$  Further, $\sum_{i \in \mathcal{Q}}x^i_j=x_j$ and hence $\sum_{i \in \mathcal{Q}}\D{x^i_j}{x_j}=1$.  Also, $\sum_{i \in \mathcal{Q}}\D{x^i_k}{x_j}=0\ \forall k \neq j.$   Thus, we just have
%
%
\BEQA
\D{\mathcal{C}^\star}{x_j} = F_j^\star  & \forall\ j \in \mathcal{Q}
\EEQA
and the controller equation is given as:
\BEQA\label{eqn_ctrlPrimal}
\dot{x}_j = (w_j - x_j F_j^\star ) & \forall\ j \in \mathcal{Q}.
\EEQA

Under this admission control loop, we have the following theorem.

\begin{theorem}
Under the time scale separation assumption, the mTracker system with dynamics (\ref{eqn_ctrlPrimal}) is globally asymptotically stable.
\end{theorem}
\begin{proof}
We use the following Lyapunov function
\BEQA
 Z(\vec{x}) &=& V(\hat{\vec{x}}) - V(\vec{x}) \label{eqn_admCtrlLyap} \\
\mbox{where } V(\vec{x})&=&\sum_{j=1}^Q w_j\log{x_j} - \mathcal{C}^\star(\vec{x}) \label{eqn_admCtrlV}
\EEQA
which is strictly concave, with $\hat{\vec{x}}$ is its unique maximum.  Differentiating $Z(\vec{x})$ we obtain
\BEQA\label{eqn_Zdot}
\dot{Z} = - \sum_{j=1}^Q \D {V}{x_j}\dot{x_j}.
\EEQA
Then from (\ref{eqn_admCtrlV}) and (\ref{eqn_ctrlPrimal})
\BEQA
\D {V}{x_j} = \frac{w_j}{x_j} - \D{\mathcal{C}^\star(\vec{x})}{x_j} &= \frac{\dot{x}_j}{x_j} \\
\therefore\ \dot{Z} = -\sum_{j=1}^Q \frac{\dot{x}_j^2}{x_j} \leq 0 & \forall\ \vec{x},
\EEQA
with $\dot{Z} = 0$ when the system is in equilibrium. Thus, the system is globally asymptotically stable \cite{khalil96}.
\end{proof}

Finally, we note that the equilibrium conditions of the controller (\ref{eqn_ctrlPrimal}) are the same as the KKT conditions of the convex optimization problem (\ref{eqn_relaxed}).  Hence, the controller succeeds in maximizing the required net utility.

\section{Matlab Simulations}\label{sec:matlab_sims}
Figure~\ref{fig:mTrack2Payoff} shows the per unit payoffs, corresponding to T2.  As expected, the per unit payoffs converge to identical values.
\begin{figure}[htbp]
\vspace{-0.15in}
\begin{center}
\includegraphics[width=\figSize]{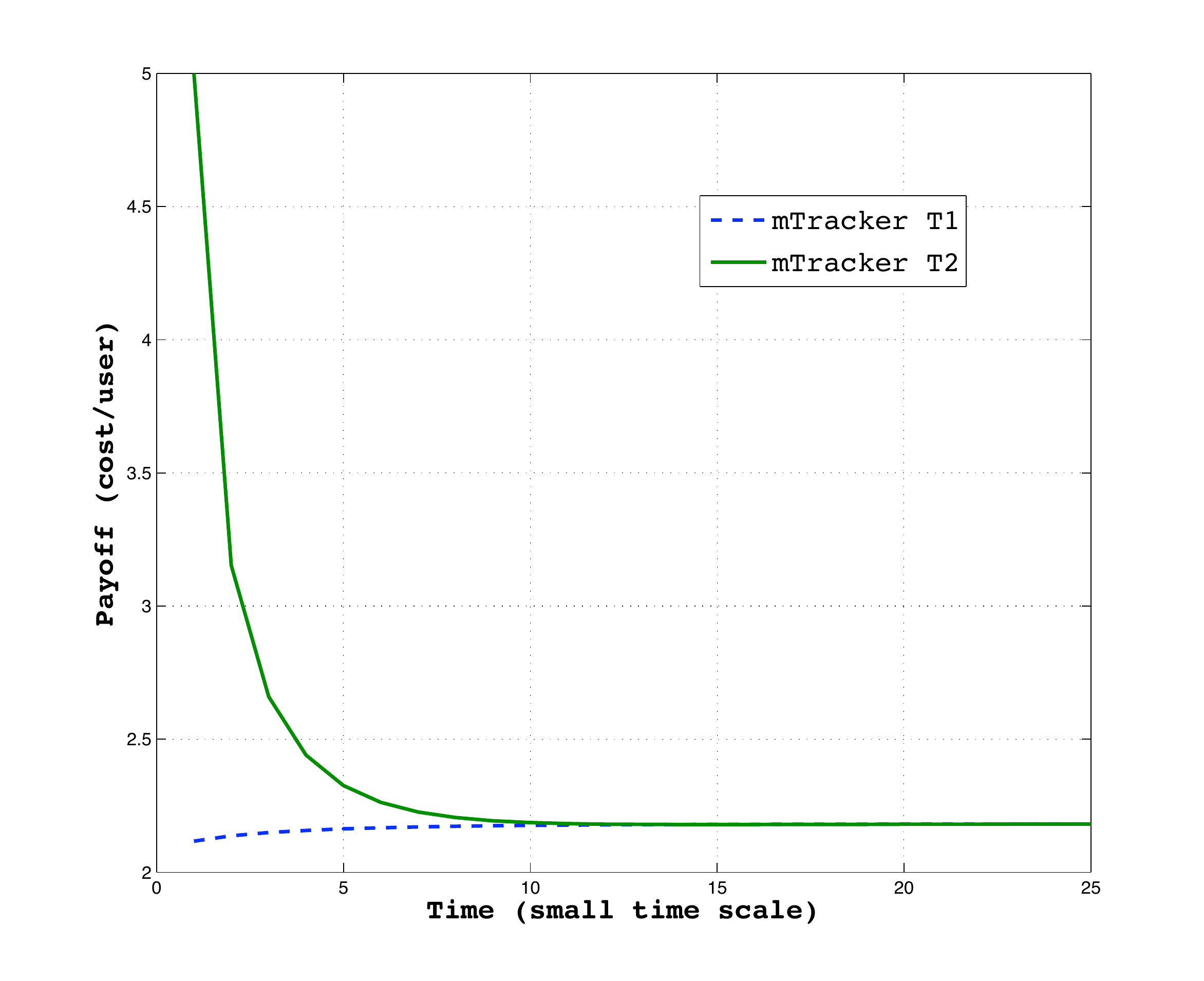}
\figCaptionSpace
\caption{The trajectory of payoffs of mTracker T2 for the 2 options available (local swarm and T1's swarm).  The payoffs eventually equalize, showing that a Wardrop equilibrium has been attained.}
\label{fig:mTrack2Payoff}
\end{center}
\figSpace
\end{figure}
We perform simulations on the simple overlay topology illustrated in Figure~\ref{fig:multitrack}.  Our objective is to validate our analytical results, and use the resulting insights to construct a realistic ns-2 implementation in the next section.  Our system consists of $3$ mTrackers (T1,T2 and T3).  The mTracker-T1 is assumed to be in steady state (i.e. it has more capacity than demand in its peer swarm) and the other mTrackers T2 and T3 are in a transient state.   Thus, T2 and T3 may forward traffic to T1.  Our simulation parameters are chosen as follows.  The initial arrival rates at the mTrackers are $x_1=10$ users/time, $x_2=20$ users/time and $x_3=20$ users/time, while the available capacities (fixed) are $C^1=30$ users/time, $C^2=20$ users/time and $C^3=20$ users/time, respectively.  There is a transit price for traffic forwarding between mTrackers, and these values are chosen as $p^1_2 = 2$ unit and $p^1_3 = 1$ unit.   

\begin{figure}[htbp]
\vspace{-0.1in}
\begin{center}
\includegraphics[width=\figSize]{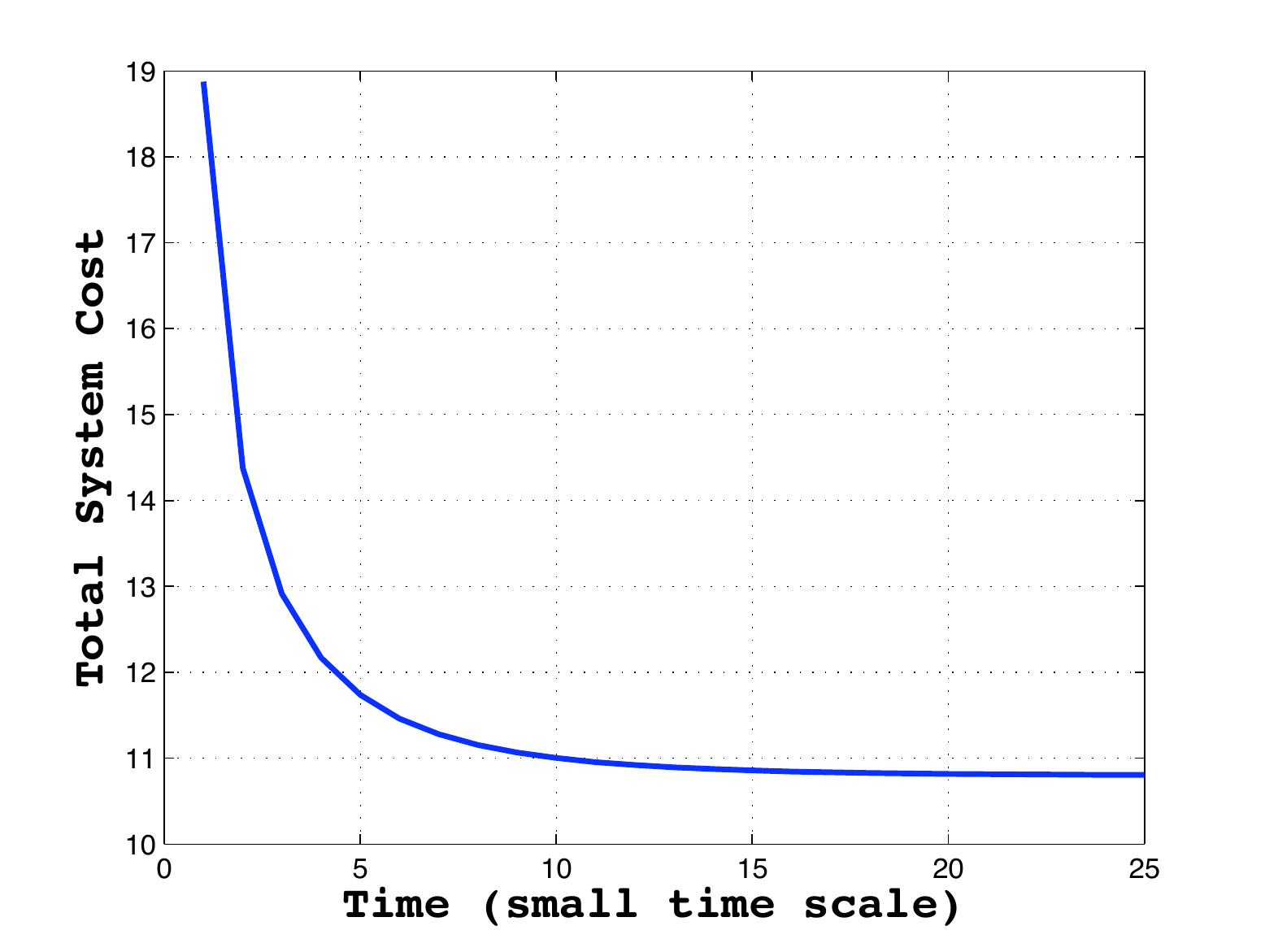}
\figCaptionSpace
\caption{The trajectory of total system cost in the system.  As expected, it decreases over time to a minimum.}
\label{fig:mTracSysDelay}
\end{center}
\vspace{-0.15in}
\end{figure}
We first validate the dynamics of the mTrackers at the small time scale.  Hence, the arrival rate at each mTracker remains fixed, and as in Section~\ref{sec:mTracker} and they each use replicator dynamics in order to balance their payoffs among available options.  We expect that (i) the per unit payoff for all available options to an mTracker should eventually be equal, and (ii) the total cost of the system would decrease to a minimum.  
We observe similar convergence for mTracker T3.  We then plot the trajectory of total system cost $\mathcal{C}(\sysState)$\footnote{Recall that this is the sum of total delay plus total transit cost.} in Figure~\ref{fig:mTracSysDelay}.  As expected it decreases with time, and converges to a minimum. 

\begin{figure}[htbp]
\vspace{-0.15in}
\begin{center}
\includegraphics[width=\figSize]{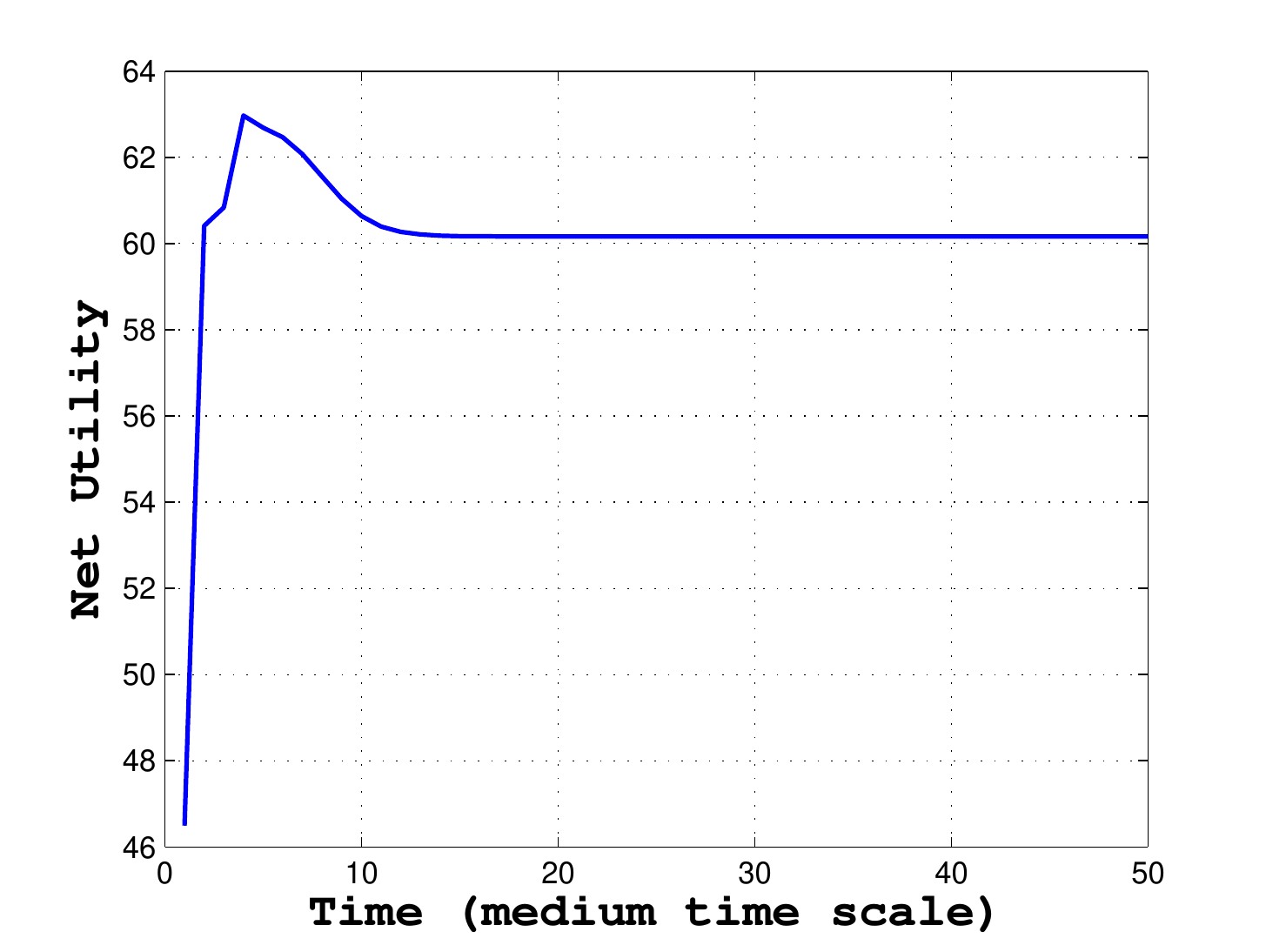}
\caption{The trajectory of net utility of the system when mTracker uses admission control.  The net utility converges to a maximum.}
\label{fig:tTracNetUtil}
\end{center}
\vspace{-0.15in}
\end{figure}
We performed simulations at the medium time scale for the admission control module and observed that the net utility of the system (as defined in (\ref{eqn_relaxed}) converges, shown in Figure~\ref{fig:tTracNetUtil}.

While our Matlab simulations suggest that our system design is valid, they do not capture the true P2P interactions within each peer-cloud.  In the next section, we implement MultiTrack using ns-2 in order to experiment with a more accurate representation of the system.

\section{ns-2 Experiments} \label{sec:ns2_sims}

We implemented the MultiTrack system on ns-2 to observe its performance in a more realistic setting.  Again, we use the same network shown in  Figure \ref{fig:multitrack}, with $3$ \textit{mTrackers} T1,T2 and T3.  However, we now explicitly model peer behavior using a BitTorrent model.  We use a flow level BitTorrent model developed by Eger \textit{et al.}, \cite{egeretal_ns2}\footnote{Here only flows are simulated, and the actual dynamics of transport protocol, like TCP, and network protocol, like IP, are ignored in the interest of lowered simulation time.}, and each peer leaves the system after completing its download.  We extended the existing BitTorrent Tracker model to support mTracker functionality.  


We estimate the delay and congestion price at each mTracker during every small time scale as follows:
\begin{enumerate}
\item {\bf Delay}: The per unit delay in each mTracker's peer cloud is measured by calculating the average download rate obtained by the peers in the current time slot, including the peers that finished service during this time slot. The delay experienced is the reciprocal of this download rate. We maintain an exponential moving average of the delay with $75\%$ weight to the delay in current time slot and $25\%$ weight to the previous value
\item {\bf Congestion Price}: The congestion price with delay $D$ and arrival rate $z$ is given as $ \D{D}{z}\times z,$ which follows from the \emph{elasticity} (\ref{eqn:elasticity}). We measure the change in delay and change in arrival rate from the previous and current time slot to calculate the congestion price.
\end{enumerate}

In our simulation, each peer has an upload capacity of $3000 \ kbps$ and their download capacity is not restricted. The requested file size is $5 \ MB$ and each chunk has a size of $256 \ kB$.  Peer arrivals are created according to Poisson processes of different rates.  T1 has 200 seeds in its peer swarm while T2 and T3 have 5 seeds each. We fix the initial arrival rates to be $x_1 =3$ users/sec, $x_2=5$ users/sec and $x_3=7$ users/sec, set transit costs to be $p^1_2=20$ and $p^1_3=10$. 
\begin{figure}[htbp]
\vspace{-0.15in}
\begin{center}
\includegraphics[width=\figSize]{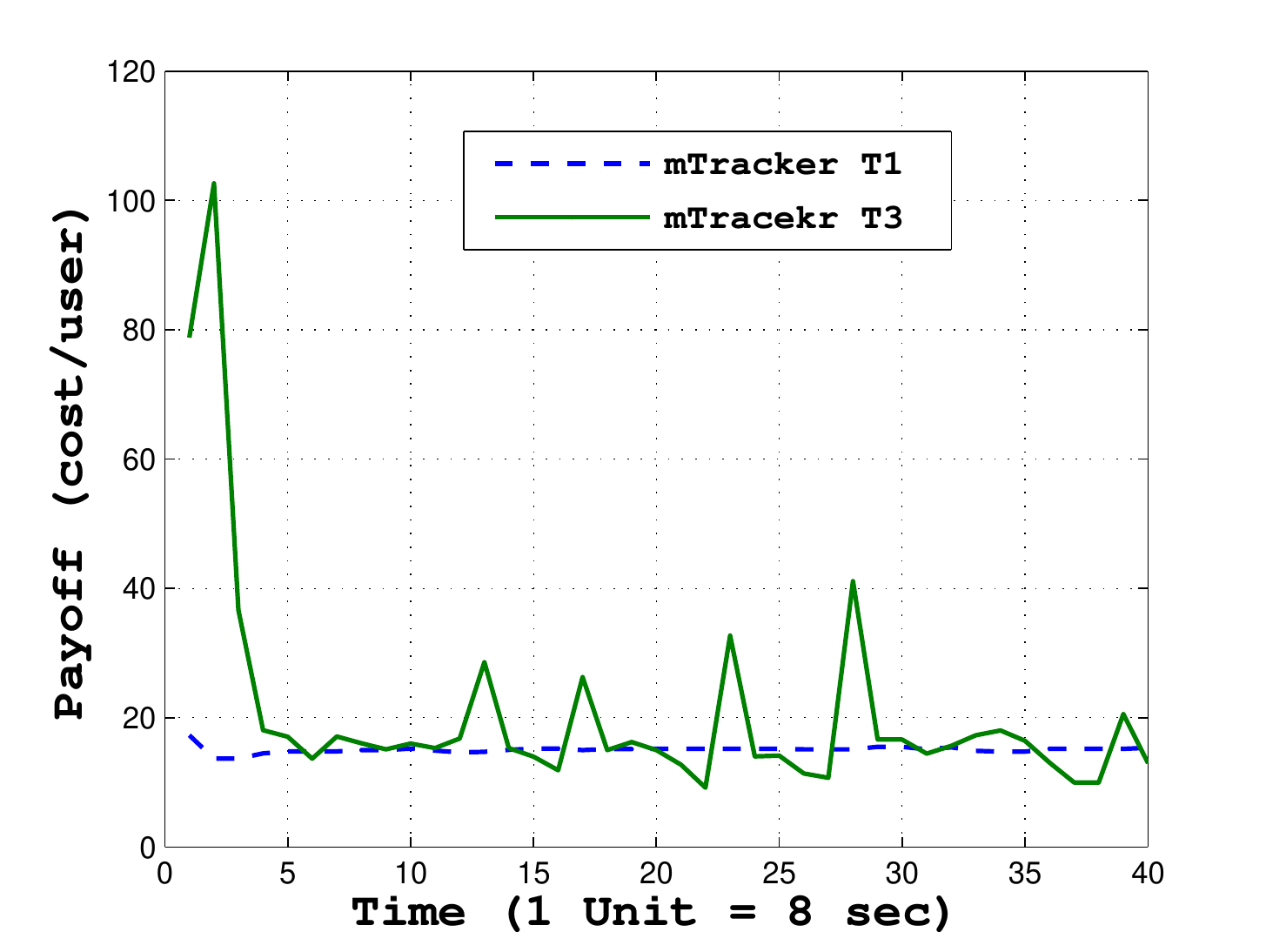}
\figCaptionSpace
 \caption{The trajectory of payoffs of mTracker T3 for the 2 options available (local swarm and T1's swarm).  The payoffs eventually equalize, showing that a Wardrop equilibrium has been attained.\label{fig:ns_mTrack_T3}}
\end{center}
\vspace{-0.17in}
\end{figure}

We first simulate the mTracker with admission control disabled so as to show the convergence of our mTracker traffic management module. We set the update interval for the mTracker to be $8$ \textit{sec}.  Thus, each mTracker calculates the splitting probabilities for the different options at this frequency. We simulate the system for $320$ \textit{sec}.

First we show the payoff convergence of the transient mTrackers. Figure~\ref{fig:ns_mTrack_T3}, shows the convergence of payoffs of mTracker T3 for its two options, local swarm and T1's swarm, thus showing that the system attains Wardrop equilibrium. We observed similar payoff convergence for other mTrackers.

Next we plot the total cost (transit price and delay cost) of our MultiTrack system. The temporal evolution of cost is shown in Figure \ref{fig:ns2:low_cost_sn}.  The impact of using MultiTrack is clearly illustrated here.  The system without traffic splitting has a high cost due to increased user delays, while traffic splitting without regard to prices has a high cost due to excessive transit traffic.  MultiTrack takes both transit price and user delay into account, and hence achieves the lowest possible cost.  

We implemented the mTracker's admission control module in ns-2. Here, at each time step the mTracker decides the admission rate (based on the dynamics developed in Section \ref{sec:mTrAdmission}).  Admission control is done at $40$ \textit{sec} intervals, while the traffic management module is run at $8$ \textit{sec} intervals during this interval.  
We simulate the system for $2000$ \textit{sec}. We expect that the net utility of the system (as defined in (\ref{eqn_relaxed})) would increase to a maximum, which is what we observe in Figure~\ref{fig:ns2_netUtil}.

\begin{figure}[htbp]
\vspace{-0.2in}
\begin{center}
\includegraphics[width=\figSize]{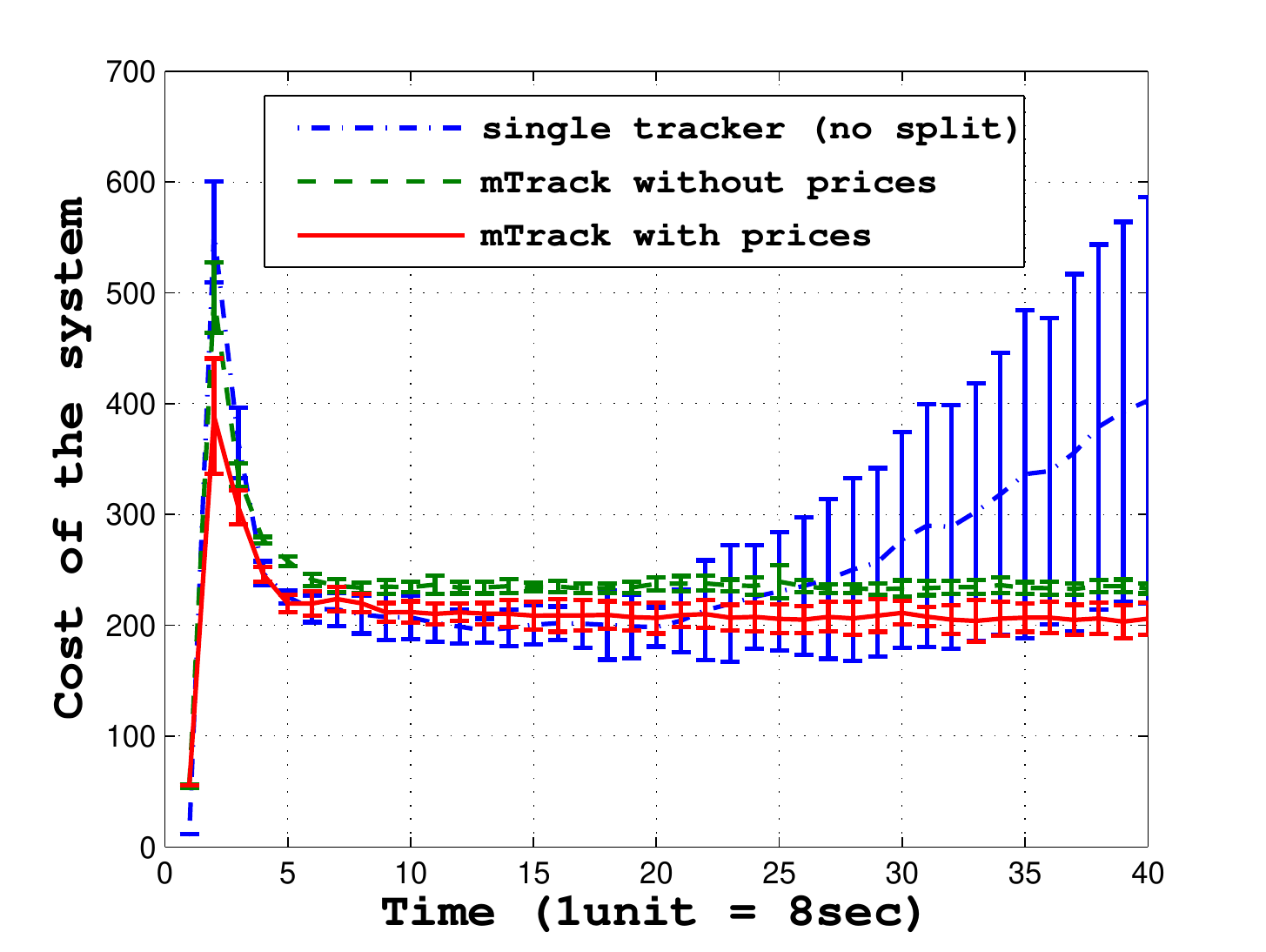}
\figCaptionSpace
 \caption{The trajectory of total system cost.  Without traffic splitting, the cost (delay plus transit price) is high.  With traffic splitting without regard to price, the delay is low but transit price is high, causing high cost.  MultiTrack takes prices and delays into account, and has lowest total cost.}
\label{fig:ns2:low_cost_sn}
\end{center}
\vspace{-0.2in}
\end{figure}
%
\begin{figure}[htbp]
\vspace{-0.2in}
\begin{center}
\includegraphics[width=\figSize]{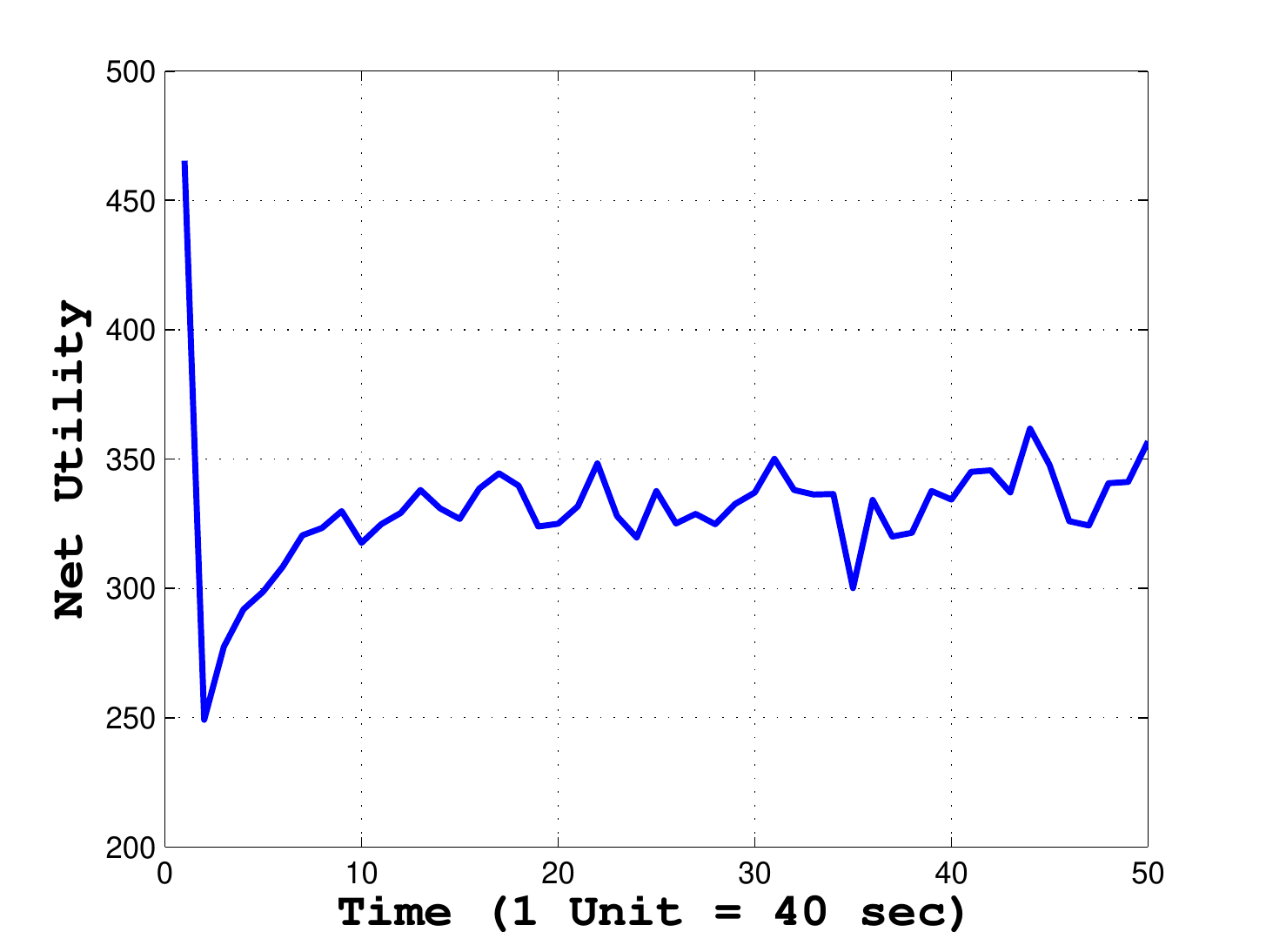}
\figCaptionSpace
 \caption{The trajectory of net utility of the system when mTracker uses admission control.}
\label{fig:ns2_netUtil}
\end{center}
\figSpace
\end{figure}

In Figure~\ref{fig:ns_AC_arrv} we can see the convergence of arrival rates of each mTracker, thus finding the optimum arrival rate into each mTracker for a fixed capacity.  Since all mTrackers have identical utility, we see that T1 dominates as its price to access its own (resource rich) swarm is zero.  Finally, we note that changing the time scales for faster responses does not seem to unduly impact stability.  In particular, reducing the small time scale from $8$ to $4$ \emph{sec} does not appreciably change our results.

\begin{figure}[htbp]
\vspace{-0.045in}
\begin{center}
\includegraphics[width=\figSize]{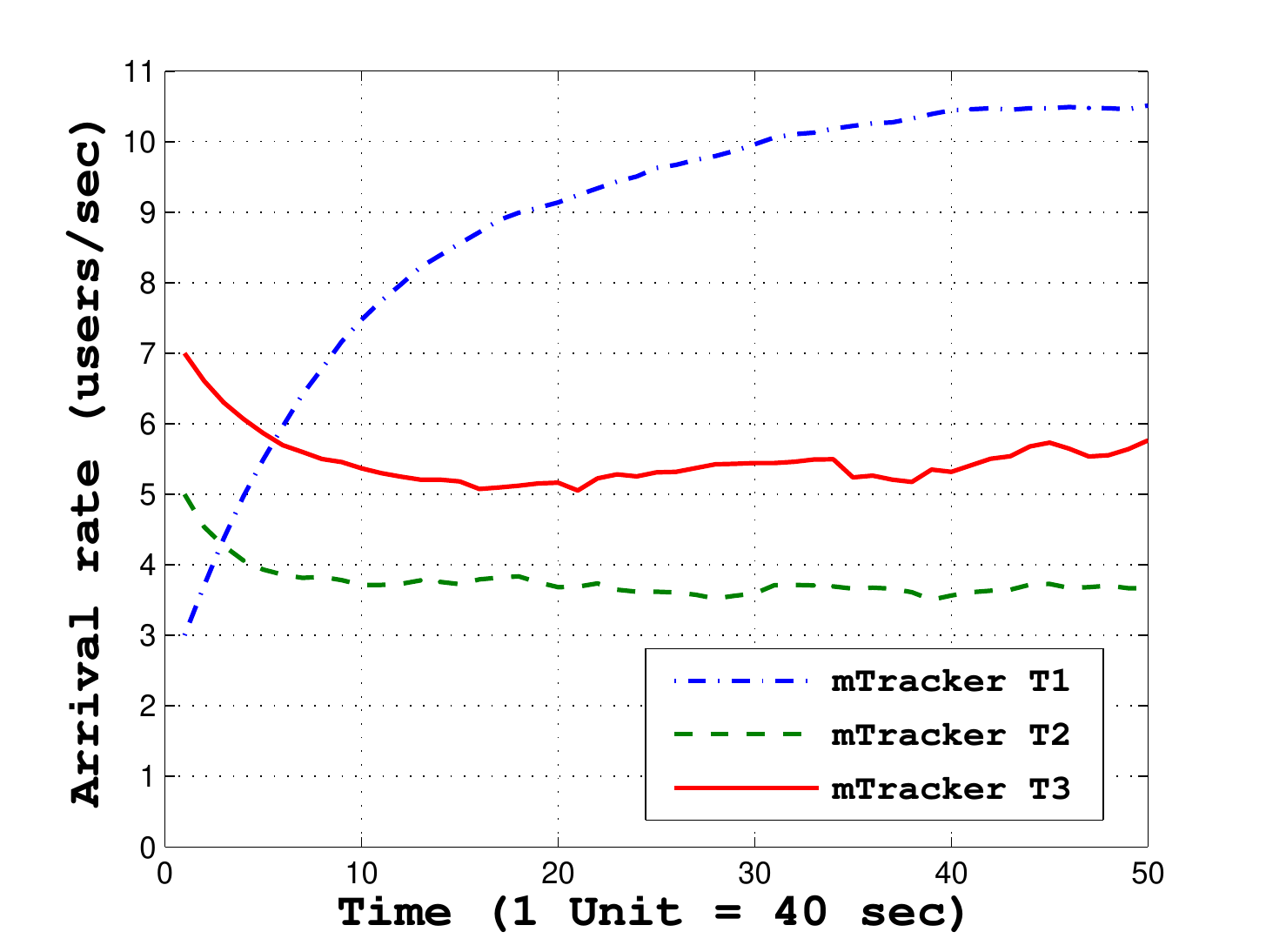}
\figCaptionSpace
 \caption{The trajectory of arrival rates for all mTrackers. The utility of each mTracker is weighed equally to a value of $10$.}\label{fig:ns_AC_arrv}
\end{center}
\vspace{-0.3in}
\end{figure}

\section{Conclusions}\label{sec:conclusion}

As the popularity of P2P systems has grown, it has become clear that aligning incentives between the system performance in terms of the user QoS, and the transit costs faced by ISPs will be increasingly important.  Fundamental to this problem is the realization that resources may be distributed geographically, and hence the marginal performance gain obtained by accessing a resource is offset in part by the marginal cost of transit in accessing it.  In this paper, we consider delay and transit costs as two dimensions and attempt to design a system---MultiTrack--that attains an optimal operating point.

Our system consists of mTrackers that form an overlay network among themselves and act as gateways to peer-clouds.  The load balancing module takes decisions based on whether the marginal decrease in delay obtained by forwarding a user to a resource rich peer-cloud is offset by the marginal increase in its transit cost.  We show that a simple price-based controllers can ensure that the total system cost is minimized in spite of each mTracker being selfish. The admission control module calculates the tradeoff between the marginal utility in increasing the admission rate in a particular ISP domain to the marginal increase in system cost, to decide the admission rates into that ISP domain.  It thus allows the correct arrival rate of users into the system to attain optimal performance.


We validated our system design using Matlab simulations, and implemented the system on ns-2 to conduct more realistic experiments.  We showed that our system significantly outperforms a system in which costs are the only control dimension (localized traffic only).  In the future, we will conduct testbed experiments on MultiTrack in a real-world setting.




%

\bibliographystyle{IEEEtran}
\bibliography{p2p}  

\end{document}